\documentclass[10pt,letterpaper,journal]{IEEEtran}

\usepackage[sort]{cite}

\usepackage[dvips]{graphicx}
\usepackage{psfrag}
\DeclareGraphicsExtensions{.eps}
\usepackage[caption=false,font=footnotesize]{subfig}
\usepackage{fixltx2e}
\usepackage{color}
\usepackage{pgfplots}
\pgfplotsset{compat=newest}
\pgfplotsset{plot coordinates/math parser=false}
\newlength\figureheight
\newlength\figurewidth
\usetikzlibrary{plotmarks}
\usepackage{morefloats}

\usepackage[cmex10]{amsmath}
\usepackage{amsfonts}
\usepackage{bbm}
\usepackage{upgreek}
\usepackage{amsmath,amssymb,amsfonts,amsbsy}
\usepackage{mathtools}
\usepackage{bbm}

\usepackage[lined,ruled,linesnumbered]{algorithm2e}
\usepackage{listings}

\usepackage{array}
\usepackage{acronym}
\usepackage{verbatim}

\newtheorem{proposition}{Proposition}

\begin{document}

\title{A Control-Theoretic Approach to Adaptive Video Streaming in Dense Wireless Networks}
\author{\IEEEauthorblockN{Konstantin Miller\IEEEauthorrefmark{1},
Dilip Bethanabhotla\IEEEauthorrefmark{2},
Giuseppe Caire\IEEEauthorrefmark{1}, and
Adam Wolisz\IEEEauthorrefmark{1}}\\
\IEEEauthorblockA{\IEEEauthorrefmark{1}Technische Universit\"at Berlin, Germany\\
Email: \{konstantin.miller, caire, adam.wolisz\}@tu-berlin.de}\\
\IEEEauthorblockA{\IEEEauthorrefmark{2}University of Southern California, Los Angeles, CA, USA\\
Email: bethanab@usc.edu }
}
\maketitle

\begin{abstract}
Recently, the way people consume video content has been undergoing a dramatic change. Plain TV sets, that have been the center of home entertainment for a long time, are losing grounds to Hybrid TV's, PC's, game consoles, and, more recently, mobile devices such as tablets and smartphones. The new predominant paradigm is: watch what I want, when I want, and where I want. 

The challenges of this shift are manifold. On the one hand, broadcast technologies such as DVB-T/C/S need to be extended or replaced by mechanisms supporting asynchronous viewing, such as IPTV and video streaming over best-effort networks, while remaining scalable to millions of users. On the other hand, the dramatic increase of wireless data traffic begins to stretch the capabilities of the existing wireless infrastructure to its limits. Finally, there is a challenge to video streaming technologies to cope with a high heterogeneity of end-user devices and dynamically changing network conditions, in particular in wireless and mobile networks. 

In the present work, our goal is to design an efficient system that supports a high number of unicast streaming sessions in a dense wireless access network. We address this goal by jointly considering the two problems of wireless transmission scheduling and video quality adaptation, using techniques inspired by the robustness and simplicity of \acf{PID} controllers. We show that the control-theoretic approach allows to efficiently utilize available wireless resources, providing high \ac{QoE} to a large number of users.
\end{abstract}

\section{Introduction}

Not so long ago, video content was largely distributed through a small number of broadcast television stations. Viewers passively consumed prescheduled content on a single available device - the television set. Things began to change in the seventies and eighties with the appearance of affordable video recorders that allowed to time-shift video programming, decoupling it from broadcasters' schedules. Video recorders also enabled consumption of prerecorded movies available through rental or sale. Meanwhile, few decades later, we are starting to observe a dramatic shift of video consumption from plain TV sets, that have been dominating home entertainment for half a century, to devices such as Hybrid TV's, personal computers, game consoles, and, more recently, mobile devices such as tablets and smartphones. This shift is accompanied by a new mindset: watch what I want, when I want, and where I want. 

The implications of this shift are manifold. Due to the asynchronicity of viewing patterns, broadcast-based distribution schemes need to be extended or replaced by multicast and unicast schemes, while the latter need to scale to millions of simultaneous viewers. In addition, an increasing fraction of the video content is distributed via wireless networks, including mobile networks, leading to a dramatic increase of wireless data traffic. Thus, traffic from wireless and mobile devices will exceed traffic from wired devices by 2018, accounting for 61 percent of the total Internet traffic. And by far the largest part of it is video. Globally, video traffic will be 79 percent of all consumer Internet traffic in 2018, up from 66 percent in 2013~\cite{CiscoVNI2014}. This development stretches the capabilities of the existing wireless infrastructure to its limits. Moreover, due to the extreme variability of network conditions on wireless and especially mobile networks, video streaming technologies are confronted with a challenge to deal with a high heterogeneity not only of end-user devices but also of dynamically changing network conditions~\cite{Mohan1999, Vetro2005a}.


It is well understood that the current trend of cellular technology (e.g., \ac{LTE}~\cite{Sesia2009}) cannot cope with the traffic increase caused by various new video services, unless the density of the deployed wireless infrastructure is increased correspondingly. In fact, throughout the history of wireless networks, throughput gains resulting from increased network density exceeded the gains from individual other factors by an order of magnitude~\cite{Chandrasekhar2008}. This motivates the recent flurry of research on massive and dense deployment of base station antennas, either in the form of "massive \ac{MIMO}" solutions (hundreds of antennas at each cell site)~\cite{Rusek2012} or in the form of very dense small-cell networks (multiple nested tiers of smaller and smaller cells, possibly operating at higher and higher carrier frequencies)~\cite{Chandrasekhar2008}. If supplied with sufficient storage capacity, these technologies can also help reducing the load on the backhaul (i.e., the wired network connecting the access network to the Internet), which have recently become a bottleneck in cellular networks~\cite{Golrezaei2013}.


At the same time, streaming over best-effort networks, that were not designed to provide stable \ac{QoS}, especially over wireless networks, makes it highly inefficient or entirely impossible
to use the same representation of a video for the duration of a streaming session. Instead, it must be adapted 
to dynamically varying network conditions such as throughput, packet loss rate, and delay jitter. 
A user might, e.g., experience continuous throughput fluctuations ranging from tens of kilobits to tens of megabits per second. 
Even a static user in an indoor residential or office \ac{WLAN} is typically exposed to interference, cross-traffic, and fading effects. The link quality fluctuations are even stronger in the case of mobile users. Thus, it is necessary to continuously adapt the \ac{QoS} requirements of the video in order to achieve a satisfactory \ac{QoE}. Recent studies suggest that these challenges have not yet been successfully addressed. In 2013, around 26.9\% of streaming sessions on the Internet experienced playback interruption due to rebuffering, 43.3\% were impacted by low resolution, and 4.8\% failed to start altogether~\cite{ConvivaReport2014}.

In the present study we focus on designing an efficient mechanism to support a high number of parallel streaming sessions in a wireless access network, such as a small-cell network%
. In order to maximize performance, we jointly consider the two problems of wireless transmission scheduling and video adaptation.
We assume \ac{HAS} as streaming technology~\cite{Stockhammer2011a, DASH2012}, which is a client-controlled adaptive streaming approach, that uses \ac{HTTP} and \ac{TCP} as its application layer and transport layer protocols, respectively, and that has been deployed with increasing success in various large-scale streaming solutions.


Inspired by the analytical tractability of \ac{PID} controllers, complemented by their powerful features, we aim at designing a joint transmission scheduling and video adaptation mechanism that resembles dynamics of a system under \ac{PID} control. A \ac{PID} controller is typically used to stabilize a dynamic system around a given target state. Its strength lies in the ability to do so in the presence of model uncertainties (that is, the system parameters are not completely known and might be time-varying) and disturbances (unknown, potentially random, inputs to the system). In our case, we use it to stabilize users' playback buffers around certain target values, in the presence of dynamically changing network conditions due to users arrival and departures, mobility and fading effects.


We evaluate the developed approach by means of simulations in different deployment scenarios, such as long-term users with low user churn, short-term users with high user churn, and a mix of short-term and long-term users. We use \ac{QoE}-related performance metrics such as total rebuffering time, start-up delay, average media bit rate, media bit rate fluctuations, and media bit rate fairness.

The rest of the paper is structured as follows. Section~\ref{sec:related_work} presents the related work. Section~\ref{sec:system_description} describes the streaming model and the wireless network model underlying our study. Section~\ref{sec:approach} presents our approach to joint scheduling and quality selection. 
Section~\ref{sec:evaluation} introduces the evaluation setting and results. Finally, Section~\ref{sec:conclusion} concludes the paper.

\section{Related work}\label{sec:related_work}

Over the past few years, adaptive video streaming has gained a lot of attention from the research community. The focus is typically on optimizing user's \ac{QoE}, complemented by a differently weighted mix of fairness, resource efficiency, energy and cost factors.

Many studies adopt the perspective of a streaming client, while modeling the network environment as a black box~\cite{Zhu2013, Liu2012, Evensen2012, MillerK2012, Jiang2012, Liu2012a, He2014}. In contrast to these studies, we jointly optimize video quality selection and network resource allocation in a dense wireless networks, in a distributed way. 


A number of studies specifically focus on video transmission over wireless networks, leveraging cross-layer techniques that jointly perform video quality selection and network resource allocation, for different types of wireless networks. In most such studies, video quality selection is performed in a centralized way~\cite{Guo2004, Djama2008, Cicalo2014, Chen2013b, Li2013b, Saki2015a, Hu2013}. In contrast, we assume a client-driven approach, where every streaming client performs adaptation individually and asynchronously w.r.t. other clients, which is inline with the \ac{HAS} streaming model~\cite{Stockhammer2011a, DASH2012}. Moreover, while these studies focus on a setting with a single base station, we consider a small-cell network~\cite{Chandrasekhar2008} with a dense base station deployment and a bandwidth reuse factor of 1. This setting is considered one of the candidate solutions to cope with the recently observed dramatic increase of wireless and mobile traffic.

Streaming technologies using unreliable transport protocols have received a lot of attention in the past~\cite{Djama2008, Cicalo2014, Saki2015a, Huang2009}. In contrast, our study assumes \ac{TCP} as transport protocol, which is reliable and thus requires a different \ac{QoE} model since, on the one hand, it prevents packet losses but, on the other hand, exhibits increased throughput fluctuations due to built-in congestion avoidance and congestion control mechanisms~\cite{Yan2012}. Moreover, due to \ac{TCP}'s end-to-end, connection oriented semantics, it does not allow in-network manipulation of packets belonging to a particular video stream.

Further studies focus on adapting the playout rate instead of video quality, in order to deal with dynamically varying network conditions~\cite{Li2006a}.


Using control-theoretic models in adaptive streaming is not a new idea. In~\cite{Zhu2013}, e.g., the authors design a video adaptation strategy which is based on a \ac{PI} controller, with the goal to combat video quality oscillations they observe when multiple video clients share a common bottleneck. In~\cite{Huang2009}, a controller with only a proportional component is used to stabilize a \ac{UDP}-based streaming system. Likewise, our work is inspired by the simplicity and power of a \ac{PID} controller. Differently, however, we treat both transmission scheduling in the wireless access network and video adaptation as a single system, which has the potential to further improve the performance. In addition, we deploy an anti-windup scheme to stabilize the controller despite of saturation, a different discretization approach, and several heuristics aiming at further improving users' \ac{QoE}.



Similar in spirit to our work, in~\cite{Bethanabhotla2015}, the authors address the goal of designing a joint transmission scheduling and video rate adaptation scheme using a \ac{NUM} framework, using the drift-plus-penalty method. The approach maximizes the sum of users' utilities which are functions of the long-term average video qualities. A drawback of this approach is the coupling of transmission scheduling slot duration and video segment size, requiring large start-up delays to combat buffer underruns.

Another hot and challenging topic is \ac{QoE} for adaptive video streaming, see e.g.~\cite{Yim2011, Narwaria2012a, Seufert2013, Reiter2014, Song2014}. Quality adaptation, rebuffering, start-up delay, and quality fluctuations are factors that have not been part of traditional \ac{QoE} metrics for video, but that have a tremendous impact on user's perception of adaptive video streamed over a best-effort network, such as the Internet. User engagement is another important metric, which is especially of interest for content providers since it is directly related to advertising-based revenue schemes~\cite{Balachandran2013}.

\section{System description}\label{sec:system_description}

In this section we describe the streaming model and the wireless network model used in our study.

\subsection{Streaming model}

We consider a set of users $u\in U_t$ at time $t$, where each user wants to watch a video file from a library of possible files. Corresponding to the \ac{HAS} model~\cite{Stockhammer2011a}, each video file is segmented in chunks of $\tau$ seconds duration, and each segment is encoded in several representations, each representation providing a different average encoding bit rate. Further, each segment starts with a random access point of the stream, thus allowing a video client to concatenate segments from different representations during the playback. A video client sequentially issues HTTP GET or GET RANGE requests to download individual segments. The meta information about available segments and representations is downloaded by the client prior to starting the streaming session, e.g., in form of an XML file, called Manifest or \ac{MPD}. An example for this technology is the open standard MPEG-DASH~\cite{DASH2012}. 

If a segment is not downloaded until its playback deadline, a buffer underrun occurs, typically followed by a rebuffering period, during which the playback is halted. In addition, a video client might initially delay the start of the playback in order to prebuffer a certain amount of video (start-up delay).

The goal of a video client is to maintain a high \ac{QoE} by adapting the media bit rate to the available network throughput. Among the main factors influencing the \ac{QoE} are the (1) amount of time spent in rebuffering and the distribution of rebuffering event over the streaming session, (2) the average and minimum video quality during the streaming session, (3) the number and distribution of quality changes, and (4) the duration of the start-up delay. The exact nature of \ac{QoE} for adaptive streaming is an ongoing research problem, see also Section~\ref{sec:related_work}.

\subsection{Wireless network model} 

We consider a wireless network with multiple user nodes and base stations sharing the wireless resource. The base stations either store cached video files, or have a (wired or wireless) connection to some video server, which we assume not to be the system bottleneck. In general, some user nodes might also participate in distribution of video data; therefore, we generally refer to base stations and user nodes serving as streaming sources as "helpers".

The network is defined by a bipartite graph $G_t=(U_t, H, E_t)$, where $t$ is the time index, $U_t$ denotes the set of users, $H$ denotes the set of helpers, and $E_t$ contains edges for all pairs $(h, u)\in H\times U_t$ for which there exists a potential transmission link between $h$ and $u$ at time $t$. We denote by $\mathcal{N}_t(u)=\left\{h\in H:(h,u)\in E_t\wedge c_{hu}(t)\geq c_{\min}\right\}$ the neighborhood of user $u$ at time $t$. Similarly, $\mathcal{N}_t(h)=\left\{u\in U_t:(h,u)\in E_t\wedge c_{hu}(t)\geq c_{\min}\right\}$. In the following, we omit the index $t$ to simplify the notation.

Although the approach presented in this paper works with different kinds of wireless network models, we will focus on the wireless network model used in~\cite{Bethanabhotla2015}, as described in the following. The wireless channel for each link $(h,u)$ is modeled as a frequency and time selective underspread fading channel~\cite{Tse2005}. Using \ac{OFDM}, the channel can be converted into a set of parallel narrowband sub-channels in the frequency domain (subcarriers), each of which is time-selective with a certain fading channel coherence time. The small-scale Rayleigh fading channel coefficients can be considered as constant over time-frequency "tiles" spanning blocks of adjacent subcarriers in the frequency domain and blocks of \ac{OFDM} symbols in the time domain. 


We assume that transmission scheduling decisions are made according to the underlying PHY and MAC air interface specifications. For example, in an LTE setting, users are scheduled over resource blocks which are tiles of 7 OFDM symbols x 12 subcarriers, spanning (for most common channel models) a single fading state in the time-frequency domain. 

Nevertheless, it is unreasonable to assume that the rate can be adapted on each single resource block. As a matter of fact, rate adaptation is performed according to some long-term statistics that capture the large-scale effects of propagation, such as distance dependent path loss and interference power. It is well-known that with a combination of rate adaptation and hybrid ARQ, a link rate given as the average with respect to the small-scale fading of the instantaneous rate function, for given large scale pathloss coefficients and interference power, can be achieved. This averaging effect with respect to the small scale fading is even more true when massive MIMO transmission is used, thanks to the fact that, due to the large dimensional channel vectors, the rate performance tends to become deterministic~\cite{Rusek2012, Huh2012}. Therefore, in our treatment we shall use the link rate function given by eq.~\eqref{eq:capacity_formula}.




We assume that the helpers transmit at constant power, and that the small-cell network makes use of universal frequency reuse, that is, the whole system bandwidth is used by all helper nodes. We further assume that every user $u$, when decoding a transmission from a particular helper $h\in\mathcal{N}(u)$ treats inter-cell interference as noise. Under these system assumptions, the maximum achievable rate during the scheduling timeslot $t_i$ for link $(h,u)\in E$ is given by
\begin{equation}
\label{eq:capacity_formula}
\begin{aligned}
&c_{hu}(t_i)=\\
&W\cdot\mathbb{E}\left[\log{\left(1+\frac{P_h g_{hu}(t_i)\left|s_{hu}\right|^2}{1+\sum_{h'\in\mathcal{N}(u)\setminus\{h\}}P_{h'} g_{h'u}(t_i)\left|s_{h'u}\right|^2}\right)}\right]\,,
\end{aligned}
\end{equation}
where $P_h$ is the transmit power of helper $h$, $s_{hu}$ is the small-scale fading gain from helper $h$ to user $u$, $g_{hu}(t_i)$ is the slow fading gain (path loss) from helper $h$ to user $u$, and $W$ is system bandwidth.

Consistently with current wireless standards, we consider the case of intra-cell orthogonal access. This means that each helper $h$ serves its neighboring users $u\in\mathcal{N}(h)$ using orthogonal FDMA/TDMA. We denote by $\alpha_{hu}(t_i)$ the fraction of time/spectrum helper $h$ uses to serve user $u$ in scheduling timeslot $t_i$. It must hold $\sum_{u\in\mathcal{N}(h)}\alpha_{hu}(t_i)\leq 1$ for all $h\in H$. The throughput of user $u$ in timeslot $t_i$ is then given by $c_u(t_i)=\sum_{h\in\mathcal{N}(u)}\alpha_{hu}(t_i)c_{hu}(t_i)$. The underlying assumption, which makes this rate region achievable, is that helper $h$ is aware of the slowly varying path loss coefficients $g_{hu}(t_i)$ for all $u\in\mathcal{N}(h)$, such that rate adaptation is possible. This is consistent with rate adaptation schemes currently implemented in \ac{LTE} and IEEE 802.11~\cite{Sesia2009,Ong2011,Molisch2010}.

We assume that all helpers are connected to a cenralized network controller, which performs the scheduling decisions.



\section{Joint scheduling and quality selection}\label{sec:approach}

In this section we present a partially distributed mechanism 
consisting of two components: video quality selection, performed by each video client independently and asynchronously, 
and link transmission scheduling, performed by a centralized network controller at equidistant time intervals.


\subsection{General idea}

Consider a dynamic system described by the following equations
\begin{equation}\label{eq:dynamic_system}
\begin{aligned}
\dot{x}(t)&=f_1\left(x(t)\right)+z(t)\\
y(t)&=f_2\left(x(t)\right)\,, 
\end{aligned}
\end{equation}
where $x(t)$ is system state at time $t$, $z(t)$ is input to the system, $y(t)$ is system's output, and $f_1(\cdot)$ and $f_2(\cdot)$ are some functions. Assume that we want to stabilize system's output $y(t)$ around a certain target value $y^*$. For many systems, this can be achieved in an efficient way by setting their input using a \ac{PID} controller, that is,
\begin{equation*}
z(t) = K_p e(t) + K_d\dot{e}(t) + K_i\int_{0}^{t}e(\xi)d\xi\,,
\end{equation*}
where the error function $e(t)=y(t)-y^*$ is the deviation of system output $y(t)$ from the target value $y^*$, and $K_p$, $K_d$, and $K_i$ are controller parameters.

In our case, we define system state as the vector of (playback) buffer levels $x(t)=\big(x_u(t),u\in U\big)$. The buffer level of a user is defined as the consecutive amount of video content, measured in seconds of playback time, stored in user's playback buffer, starting from the current playback point. We denote by $r_u(t)$ the encoding bit rate, or media bit rate, of the video content downloaded by user $u$ at time $t$, measured in bit/s. Further, $c_u(t)$ shall denote the throughput experienced by user $u$ at time $t$, in bit/s. With this notation, the buffer level dynamics are represented by the following system of equations:
\begin{equation}\label{eq:buffer_dynamics}
\dot{x}_u(t)=\frac{c_u(t)}{r_u(t)}-1,\;\;\forall u\in U\,,
\end{equation}
where $\frac{c_u(t)}{r_u(t)}$ is the rate at which the buffer is filled by the arriving video data, and $-1$ is the rate at which the buffer is emptied by the playback.

Particularizing the general dynamic system equation~\eqref{eq:dynamic_system} to our specific case defined in~\eqref{eq:buffer_dynamics}, we obtain $f_1\left(x(t)\right)\equiv -1$, $f_2\left(x(t)\right)=x(t)$, and $z_u(t)=\frac{c_u(t)}{r_u(t)}$. That is, our goal is to stabilize users' buffer levels $x(t)$ around a target value $x^*$ by controlling $\frac{c_u(t)}{r_u(t)}$. In the following, to simplify notation and w.l.o.g., we set $x^*=0$.

In order to apply \ac{PID} control to buffer level dynamics~\eqref{eq:buffer_dynamics}, we have to set system's input $\frac{c_u(t)}{r_u(t)}$ to the controller output
\begin{equation}\label{eq:pid_equality}
\frac{c_u(t)}{r_u(t)}=\underbrace{K_{p,u} x(t) + K_{d,u}\dot{x}(t) + K_{i,u}\int_{0}^{t}x_u(\xi)d\xi}_{=:\omega_u(t)}
\end{equation}
for each user $u\in U$. Consequently, the buffer level dynamics will obey the following differential equation
\begin{equation}\label{eq:system_x}
\dot{x}_u(t)=\omega_u(t) -1,\;\forall u\in U\,.
\end{equation}

In the following we present a basic result on stability of system~\eqref{eq:system_x}. Proposition~\ref{prop:1} ensures that it converges to the target buffer level $x^*=0$ for arbitrary initial values.
\begin{proposition}\label{prop:1}
System~\eqref{eq:system_x} is globally asymptotically stable if and only if $K_{d,u}\neq 1$, $\frac{K_{p,u}}{1-K_{d,u}}<0$ and $\frac{K_{i,u}}{1-K_{d,u}}<0$.
\end{proposition}
\begin{proof}
In the following, we omit the user index $u$.

First, we observe that for $K_d=1$ system~\eqref{eq:system_x} degenerates and is only solvable for zero initial value: $x(0)=0$.

Assuming $K_d\neq 1$ and substituting $y(t)=\int_0^t x(\xi)d\xi$, we transform~\eqref{eq:system_x} into a system of linear differential equations 
\begin{equation*}
\begin{bmatrix}\dot{x}(t)\\\dot{y}(t)\end{bmatrix}=\begin{bmatrix}a&b\\1&0\end{bmatrix}
\begin{bmatrix}x(t)\\y(t)\end{bmatrix}-\begin{bmatrix}c\\0\end{bmatrix}\,,
\end{equation*}
with $a=\frac{K_p}{1-K_d}$, $b=\frac{K_i}{1-K_d}$, and $c=\frac{1}{1-K_d}$.

We proof the claim by explicitely constructing the general solution. We distinguish two cases. Case I: $a^2+4b=0$. In this case, the solution for the initial value problem $x(0)=x_0$ is given by
\begin{equation*}
x(t)=\big(x_0+(0.5x_0 - 1)at\big) e^{\frac{1}{2}at}\,.
\end{equation*}
It converges to 0 for $t\rightarrow\infty$ if and only if $a<0$.
For case II: $a^2+4b\neq 0$, we obtain
\begin{align*}
&x(t)=\\
&\frac{dx_0-ax_0+2c}{2d}e^{0.5(a-d)t}+\frac{dx_0+ax_0-2c}{2d}e^{0.5(a+d)t}\,,
\end{align*}
with $d=\sqrt{a^2+4b}$. It converges to 0 for $t\rightarrow\infty$ if and only if $Re(a+d)<0$ and $Re(a-d)<0$, where $Re(\cdot)$ denotes the real part of a complex number. This, however, is equivalent to $a<0$ and $b<0$, proving the claim.
\end{proof}

In real deployment scenarios, system's input $\frac{c_u(t)}{r_u(t)}$ cannot always be set according to~\eqref{eq:pid_equality} due to several reasons that we list in the following, along with references to sections where we address them.

\begin{itemize}
\item
The range of feasible values for the left-hand side of~\eqref{eq:pid_equality} is bounded by $\left[0, \frac{c_{u,\max}}{r_{u,\min}}\right]$, where $c_{u,\max}$ is the maximum throughput that can be allocated to user $u$ by the network, and $r_{u,\max}$ is the smallest available media bit rate. (The lower bound is attained if the download is paused. The upper bound is attained when the user downloads lowest video quality at maximum available throughput.) This means that controller gain values outside of this region cannot be applied to the system. Such a controller with a limited gain is called saturated. It might become unstable due to the problem called integral windup. We address this issue in Section~\ref{sec:integral_windup}.
\item
In a real deployment scenario, the system we consider is a sampled and distributed system. Both $c_u(t)$ and $r_u(t)$ cannot be adapted in continuous time. They can only be modified at certain time instants, and after that they are kept constant for a certain period of time. Moreover, both transmission scheduling and quality selection take place at different time instants, since quality selection takes place once per video segment, while transmission scheduling takes place once per scheduling time slot. We address this issue in Section~\ref{sec:discretization}.
\item
Lastly, stabilizing buffer level around a target value is not enough to ensure high \ac{QoE}. It primarily aims at avoiding buffer underruns, but, in addition, we also would like to avoid excessive video quality fluctuations and unfairness among individual users. From the perspective of buffer level stability, assigning a user low throughput and low video quality stabilizes his buffer level as good as assigning him high throughput and high video quality, whereas from the perspective of \ac{QoE} the second situation is clearly preferred (see, e.g.,~\cite{Yim2011} for the impact of quality fluctuations on \ac{QoE}). We will address these issues in Sections~\ref{sec:quality_selection} and~\ref{sec:transmission_scheduling}.
\end{itemize}

Finally, it is worth noting that it might make sense to define different target levels $x_u^*$ for different users, in order to account for their individual mobility patterns, link statistics, and \ac{QoE} expectations. We may imagine a specific per-user adaptation, e.g. at application layer, that acts on the control parameter $x_u^*$. In the present work we do not address this outer control, and assume a common target $x^*$ for all users.

Also note that in general, there are several reasons for keeping $x^*$ at a reasonably low level. Especially in the case of short videos, if a user prematurely quits the streaming session, the bandwidth used to download video data remaining in the buffer would have been wasted, engendering costs for content providers and network operators~\cite{Finamore2011,Chen2013f}. Further, if a user downloads certain parts of a video in a low quality, e.g. due to poor network throughput, he might partially discard this data when network throughput increases, leading to a waste of bandwidth as in the previous case. Finally, in the case of a live stream, video content becomes available incrementally, imposing a trade-off between the value of $x^*$ and the viewing delay.

\subsection{Integral windup}\label{sec:integral_windup}

As mentioned in the previous section, in practice, it is not always possible to set control variables $c_u(t)$ and $r_u(t)$ such that equality~\eqref{eq:pid_equality} is satisfied, due to limited link rate on the one hand, and a limited set of available media bit rates on the other hand. Our controller becomes a saturated controller.

Saturated \ac{PID} controllers have been subject of intensive research in the last decade. Partially, the attention has been motivated by control of robotic manipulators. One issue with saturated controllers is the so-called integral windup. For large deviations of $x(t)$ from its target value, the unsaturated controller would apply a high positive or negative gain to bring $x(t)$ back to $x^*$. Due to the saturation, however, only a smaller gain can be applied. Thus, it takes more time to bring the state back to the target value. During this time, however, the error integral obtains a larger value than it would have had otherwise. The result is a higher overshoot and oscillations, leading to potential instability.

To formalize the notion of saturation, we use the following notation
\begin{equation*}
\big[x\big]_{x_{\min}}^{x_{\max}}=\begin{cases}
x_{\min}\;\;\text{for}\;\;x\leq x_{\min}\\
x_{\max}\;\;\text{for}\;\;x\geq x_{\max}\\
x\;\;\text{otherwise}
\end{cases}\,.
\end{equation*}
With this definition, the saturated, and thus more realistic, version of system~\eqref{eq:system_x} can be written as
\begin{equation}\label{eq:prop_2}
\dot{x}_u(t)=\big[\omega_u(t)\big]_0^{g_{\max}}-1,\;\forall u\in U\,,
\end{equation}
with $g_{\max}>1$, and $\omega_u(t)$ as defined in~\eqref{eq:pid_equality}.

The following proposition, which leverages ideas and results from~\cite{Alvarez-Ramirez2003} and~\cite{Hoppensteadt1974}, shows that for small enough $\lvert K_{i,u}\rvert$, saturated system~\eqref{eq:prop_2} retains its global asymptotic stability property.
\begin{proposition}\label{prop:2}
Assume that conditions of Proposition~\ref{prop:1} are fulfilled. Then, for every set of initial values $\left[x_{\min},x_{\max}\right]$ with $x_{\min}\leq 0\leq x_{\max}$ there exists a $\tilde{K}_{i,u}>0$ such that for $\lvert K_{i,u}\rvert<\tilde{K}_{i,u}$ and $x(0)\in\left[x_{\min},x_{\max}\right]$ the solution of the initial value problem~\eqref{eq:prop_2} converges to $0$ for $t\rightarrow\infty$. This property is sometimes called \emph{semi-global practical stability}.
\end{proposition}
\begin{proof}
Throughout the proof, we will omit the user index $u$. First, we rewrite~\eqref{eq:prop_2} as
\begin{equation*}
\dot{x}(t)=\left[ax(t)+b\int_0^t x(\xi)d\xi\right]_{c-1}^{g_{\max}+c-1}-c\,,
\end{equation*}
where $a=\frac{K_p}{1-K_d}$, $b=\frac{K_i}{1-K_d}$, and $c=\frac{1}{1-K_d}$. Next, we substitute $y(t)=b\int_0^t x(\xi)d\xi$. We obtain the equivalent formulation
\begin{equation}\label{eq:prop_2_equiv}
\begin{cases}
\dot{x}(t)=\big[ax(t)+y(t)\big]_{c-1}^{g_{\max}+c-1}-c\\
\dot{y}(t)=bx(t)
\end{cases}
\end{equation}
Observe that the unique equilibrium of this system is $(0,c)$. In order to shift the equilibrium to $(0,0)$, we define $\tilde{x}(t)=x(t)-\frac{c-y(t)}{a}$ and $\tilde{y}(t)=y(t)-c$. We obtain
\begin{equation*}
\begin{cases}
\dot{\tilde{x}}(t)=\big[a\tilde{x}(t)+c\big]_{c-1}^{g_{\max}+c-1}-c+b\left(\frac{1}{a}\tilde{x}(t)-\frac{1}{a^2}\tilde{y}(t)\right)\\
\dot{\tilde{y}}(t)=b\left(\tilde{x}(t)-\frac{1}{a}\tilde{y}(t)\right)
\end{cases}
\end{equation*}
Next, we write the integral gain as $K_i=\epsilon\bar{K}_i$ and substitute the time variable $t'=\epsilon t$, where $\epsilon >0$. We define new variables $\chi(t')=\tilde{x}(t'/\epsilon)$ and $\zeta(t')=\tilde{y}(t'/\epsilon)$ to obtain a "fast" version of our system
\begin{equation}\label{eq:prop2_fast_version}
\begin{cases}
\epsilon\chi'(t')=\big[a\chi(t')+c\big]_{c-1}^{g_{\max}+c-1}-c+\epsilon\bar{b}\left(\frac{1}{a}\chi(t')-\frac{1}{a^2}\zeta(t')\right)\\
\zeta'(t')=\bar{b}\left(\chi(t')-\frac{1}{a}\zeta(t')\right)
\end{cases}
\end{equation}
where $\bar{b}=\frac{\bar{K_i}}{1-K_d}$. Observe that $(0,0)$ is the equilibrium point of~\eqref{eq:prop2_fast_version}. Further, observe that $\left(\chi(t'),\zeta(t')\right)\rightarrow(0,0)$ as $t'\rightarrow\infty$ implies that $\left(\tilde{x}(t),\tilde{y}(t)\right)\rightarrow(0,0)$ as $t\rightarrow\infty$, and thus, stability results for~\eqref{eq:prop2_fast_version} transfer to~\eqref{eq:prop_2}. Further, for $\epsilon\ll 1$, the system~\eqref{eq:prop2_fast_version} is in the form of standard singular perturbation~\cite{Hoppensteadt1974}.

Now, it is relatively straightforward to validate that the conditions of Theorem 3 in~\cite{Hoppensteadt1974} are fulfilled for system~\eqref{eq:prop2_fast_version}%
, proving the claim.
\end{proof}

Note that this result only guarantees stability if the integral coefficient $K_{i,u}$ is small enough. However, making $K_{i,u}$ smaller than necessary might have negative impact on convergence speed. In practice, it is difficult to compute the maximum value $K_{i,u}$ that still ensures stability. Therefore, in the following, we present an alternative approach to solve the problem of integral windup: conditional integration. With this approach, the value of the integral part of the controller is not allowed to exceed certain hard limits. 

In particular, we use the equivalent form~\eqref{eq:prop_2_equiv} of the saturated system~\eqref{eq:prop_2}, and apply a bound $g_{i,\text{max}}$ on the error integral. We obtain
\begin{equation}\label{eq:saturated2}
\begin{cases}
\dot{x}(t)=\big[ax(t)+y(t)\big]_{c-1}^{g_{\max}+c-1}-c\\
\dot{y}(t)=
\begin{cases}
\max\left(0,bx(t)\right)\;\text{if}\;y(t)\leq c-g_{i,\max}\\
\min\left(0,bx(t)\right)\;\text{if}\;y(t)\geq c+g_{i,\max}\\
bx(t)\;\text{otherwise}
\end{cases}
\end{cases}
\end{equation}
where $a$, $b$, and $c$ are as defined in the proof of Proposition~\ref{prop:2}. The resulting controller is then given by
\begin{equation}\label{eq:aw_control}
\frac{c_u(t)}{r_u(t)}=ax(t)+y(t)\,,
\end{equation}
with $x(t)$, $y(t)$ defined by~\eqref{eq:saturated2}. The advantage of this formulation is that it accounts for the saturated gain, and limits the value of the error integral, reducing the impact of integral windup. An analytic study of the performance of this anti-windup strategy is notably complex, forcing us to resort to a simulative evaluation, presented in Section~\ref{sec:evaluation}. 

Further potential anti-windup strategies include limiting the integral action of the controller from growing by keeping it constant whenever the controller enters saturation, adding anti-windup compensating terms to the integral action, etc. (see, e.g.,~\cite{Tarbouriech2009}).

\subsection{Sampled distributed system}\label{sec:discretization}

So far, we studied a system, where control decisions, that is, transmission scheduling and video adaptation, are made in continuous time. In practice, however, both control decisions take place at discrete time instances, while in between, the values of the control variables are fixed. In this section, we reformulate our approach to adapt it to this requirement.

The challenge here stems from the fact that while we may assume that transmission scheduling takes place regurlarly, at equidistant time intervals, quality selection can only take place when a user starts downloading a new video segment, which happens for each user independently and, in general, on a different time scale than transmission scheduling. 

Moreover, time intervals between individual segment downloads may be subject to considerable variation over time. Whenever the buffer level of a user is in equilibrium (that is, it stays around the target level for a certain period of time), the average duration of a segment download equals the duration of the segment itself. However, when buffer level is increasing or decreasing, the duration of a segment download might be subject to significant fluctuations. In addition, segment sizes might substantially deviate from representation averages, due to \ac{VBR} encoding used by modern compression technologies, causing further variations of download times.

Let us for the moment assume that both control decisions, scheduling and adaptation take place simultaneously at time instants $t_i$, $i=0,1,\ldots$ 
In order for the sampled system to have the same state as the continuous system at certain given time instants $t_i$, we need to set our control variables as follows:
\begin{equation}\label{eq:sampled_system}
\frac{c_u(t)}{r_u(t)}=\underbrace{\frac{x(t_{i+1}) - x(t_i)}{t_{i+1}-t_i}}_{=:\tilde{\omega}_u(t_i,t_{i+1})}+1\,,\;\forall t\in[t_i,t_{i+1}),\,
\end{equation}
where $x(t_i)$ is the buffer level at time $t_i$ and $x(t_{i+1})$ is computed by solving the initial value problem defined by~\eqref{eq:saturated2}.

\begin{proposition}\label{prop:4}
Assume that conditions of Proposition~\ref{prop:1} are fulfilled. Then, sampled control~\eqref{eq:sampled_system} and continuous time saturated control with anti-windup~\eqref{eq:aw_control} lead to identical system states at time instants $t_i$, $i\in\mathbb{N}$.
\end{proposition}
\begin{proof}
The claim is proven by substituting~\eqref{eq:sampled_system} into~\eqref{eq:buffer_dynamics}, integrating the right hand side, and using $x(t_{i+1})$ that solves the initial value problem~\eqref{eq:saturated2}.
\end{proof}


In real deployment, however, $c_u(t)$ and $r_u(t)$ cannot be set simultaneously. Instead, we are dealing with a distributed system, where transmission scheduling and quality selection are performed independently from each other and at different time instants. In the following two sections, we present heuristics for controlling transmission scheduling and quality selection in a distributed way. 

\subsection{Quality selection}\label{sec:quality_selection}

In the following, we present several heuristics that complement the mechanisms presented in previous sections, so that we obtain a distributed, practically implementable approach. While this section covers the quality selection part, the following Section~\ref{sec:transmission_scheduling} covers transmission scheduling.

The idea we use to organize operation of our controller in a distributed way is the following. The network on the one hand and each individual user on the other hand shall try to maintain the equality~\eqref{eq:sampled_system} every time they adapt their respective decision variable. The network does so at the beginning of each scheduling timeslot, while each user does so when he is about to request a new video segment.

We denote by $t_{s,u}$ the time, when a user is about to start a segment download. We transform~\eqref{eq:sampled_system} to obtain the quality selection rule
\begin{equation}\label{eq:user_update}
r_u(t_{s,u})=\frac{c_u(t_{s,u})}{1+\tilde{\omega}_u(t_{s,u},t_{f,u})}\,,
\end{equation}
where $t_{f,u}$ is the time when the segment download would finish if the buffer dynamics would obey~\eqref{eq:saturated2}. It can be computed by solving $x_u(t_{f,u})=x(t_{s,u})+\tau-(t_{f,u}-t_{s,u})$, with $x_u(t)$ being the solution to the initial value problem defined by~\eqref{eq:saturated2}. 

In addition, in order to avoid excessive quality fluctuations and provide a smooth adaptation trajectory, we use exponential moving average for quality adaptation. The network throughput $c_u(t_{s,u})$ is approximated by a simple moving average of past throughput:
\begin{equation}\label{eq:user_update_ma}
r_u(t_{s,u})=(1-\alpha)r_u(t_{s,u}^-)+\alpha\frac{<c_u(t)>_{t\in[t_{s,u}-T,t_{s,u}]}}{1+\tilde{\omega}_u(t_{s,u},t_{f,u})}\,,
\end{equation}
where $r_u(t_{s,u}^-)$ is the previously selected video quality, $\alpha\in(0,1)$, and $T$ is a configuration parameter denoting the time period over which we compute the throughput average.

Since only a finite set of media bit rates is available, we round $r_u(t_{s,u})$ down to the next available value: $\tilde{r}_u(t_{s,u})=\max\left\{r\in R_u\;|\;r\leq r_u(t_{s,u})\right\}$. If $\left\{r\in R_u\;|\;r\leq r_u(t_{s,u})\right\}$ is empty, the lowest available media bit rate is selected.

In order to account for the difference between the target media bit rate $r_u(t_{s,u})$ and the actually selected media bit rate $\tilde{r}_u(t_{s,u})\leq r_u(t_{s,u})$, we introduce a delay after the download of the current segment, computed as follows. If the media bit rate $r_u(t_{s,u})$ were actually available and assuming user's throughput $c_u(t_{s,u})$ would remain constant, the buffer level after downloading the current segment would be $x(t_{s,u}) - \left(1+\frac{r_u(t_{s,u})}{c_u(t_{s,u})}\right)\tau$. Since the actually selected media bit rate is smaller, however, the buffer level after the download will be larger. Thus, in order to account for the difference in media bit rates, we delay the subsequent request until the buffer level drops below $\max\left(0,\;x(t_{s,u}) -\tau,\;x(t_{s,u}) -\left(1+\frac{r_u(t_{s,u})}{c_u(t_{s,u})}\right)\tau\right)$, where the maximum operator is an additional precaution preventing the buffer from depleting more than the duration of one segment at a time, as well as from falling below the target level 0 as a result of a delayed request.


During a particularly long period of low throughput the buffer may become empty, interrupting the playback. We call such an event a buffer underrun. In order to avoid multiple buffer underruns within a very short period of time, the client starts to play a segment only after it has been fully downloaded. This is called rebuffering. We limit the duration of the rebuffering period to the download time of one segment, that is, once we have at least one segment in the buffer, the playback is immediately restarted.

Finally, at the begin of the streaming session when no throughput information is available, the client downloads the first segment in lowest quality in order to minimize start-up delay.



\subsection{Transmission scheduling}\label{sec:transmission_scheduling}

The goal of the network is on the one hand to provide the capacities requested by the users, that is, to maintain~\eqref{eq:sampled_system}. On the other hand, it shall allocate the remaining capacity, if available, in a fair manner in order to provide high network utilization and to eventually enable users to switch to a higher video quality. In order to achieve these goals, we let the network controller solve a series of linear optimization problems. In the following, $\tilde{r}_u(t_{s,u})$ and $\tilde{\omega}_u(t_{s,u},t_{f,u})$ are values computed by user $u$ for the last segment requested prior to scheduling timeslot $t_i$ and communicated to the network as part of the download request. Further, we denote by $\rho_u(t_{s,u})=\tilde{r}_u(t_{s,u})(1+\tilde{\omega}_u(t_{s,u},t_{f,u}))$ the throughput demand of user $u$.


First, the network controller maximizes the minimum fraction of user's throughput relative to his demand, similar to the well-known maximum concurrent flow problem. At the same time, the controller tries to improve fairness by encouraging users streaming at low video quality to switch to higher quality, if there are sufficient network resources. This is achieved by artificially raising the demand $\rho_u(t_{s,u})$ of the 10\% of the users with lowest demand to the 10th percentile across all users.
We obtain the following optimization problem:
\begin{align}
{\max}\;\;\underset{u\in U}{\min}\quad 
&\frac{c_u(t_i)}{\max\left(\rho_u(t_{s,u}),\,\rho_{10}(t_i)\right)} \tag{\bf TS1}\label{TS1} \\
\text{s.t.}\quad 
&\sum_{u\in\mathcal{N}(h)}\alpha_{hu}\leq 1\,,\quad\forall h\in H \tag{C1} \\
&\alpha_{uh}\geq 0\,,\quad\forall h\in H,\;\forall u\in U\,, \tag{C2}
\end{align}
where $\rho_{10}(t_i)$ is the 10-percentile of $\left(\rho_u(t_{s,u}), u\in U\right)$. 
Recall that $c_u(t_i)=\sum_{h\in\mathcal{N}(u)}\alpha_{hu}c_{hu}(t_i)$. 
We denote the optimum value of~\eqref{TS1} by $\vartheta^*$.

In the second step, the network controller fixes the minimum relative allocated capacity to its optimum value $\vartheta^*$, and maximizes the minimum allocated capacity.
\begin{align}
{\max}\;\;\underset{u\in U}{\min}\quad 
&c_u(t_i) \tag{\bf TS2}\label{TS2}\\
\text{s.t.}\quad 
&\vartheta^*\leq\frac{c_u(t_i)}{\max\left(\rho_u(t_{s,u}),\,\rho_{10}(t_i)\right)},\,\;\forall u\in U \tag{C3}\label{C3} \\
& \text{(C1), (C2)}\notag
\end{align}
We denote the optimum value of~\eqref{TS2} by $c_{\min}^*$.

Finally, it fixes the minimum allocated relative capacity $\vartheta^*$ and minimum allocated absolute capacity to $c_{\min}^*$ and maximizes the total network throughput. In order to avoid high-amplitude throughput spikes for the individual users, it limits the capacity allocated to a user to either twice the median demand across all users or twice the minimum allocated capacity $c_{\min}^*$, whichever is larger.
\begin{align*}
{\max}\quad 
&\sum_{u\in U}c_u(t_i) \tag{\bf TS3}\label{TS3}\\
\text{s.t.}\quad 
&c_u(t_i)\geq c_{\min}^*,\,\;\forall u\in U \tag{C4}\\
&c_u(t_i)\leq 2\max{\big(c_{\min}^*,\,\rho_{50}(t_i)\big)},\,\;\forall u\in U \tag{C5}\\
& \text{(C1), (C2), (C3)}\,,\notag
\end{align*}
where $\rho_{50}(t_i)$ is the median demand across all users.

Since the maximum number of users per helper is limited by a technology dependent value, the number of optimization variables and constraints is $\mathcal{O}\left(\max{\left(|H|,|U|\right)}\right)$, which can be handled very efficiently by modern linear program solvers even for large networks. Also note that most solvers allow to iteratively modify and reoptimize a model, which reduces the complexity of subsequent optimizations such as we have here.




\section{Evaluation}\label{sec:evaluation}

In the following, we present our evaluation setting and results. All results were obtained by means of simulations. The simulation code was written in C++, we used Gurobi~\cite{Gurobi} to solve optimization problems, and we used odeint~\cite{odeint} to solve differential equations.

Section~\ref{sec:performance_metrics} describes the performance metrics. 
Section~\ref{sec:evaluation_settings} describes the general setting, such as network and video parameters.
Section~\ref{sec:experimental_design} elaborates on the goal and design of the individual experiments. 
Finally, Section~\ref{sec:evaluation_results} presents evaluation results.

\subsection{Performance metrics}\label{sec:performance_metrics}

We use the following metrics to assess the performance of the proposed system.

\subsubsection{Stability} 
A well-known issue with closed-loop control systems is their potential to become unstable, leading to high-amplitude fluctuations of the system state. Although not necessarily harmful per se, instability can have a dramatic impact on other performance metrics. We evaluate stability by means of buffer level statistics, such as the maximum buffer level overshoot and the minimum buffer level of a user during a simulation run.

\subsubsection{Rebuffering duration} 
When a client's video buffer has been drained so that the next video segment does not arrive before his playback deadline, the playback must be halted. This is often refered to as a buffer underrun. A buffer underrun is followed by a rebuffering period, where the client waits until enough video data is accumulated in the buffer to resume playback. The conditions that need to be fulfilled before the playback is resumed depend on client's rebuffering strategy. In our design, we resume playback after at least one segment is completely downloaded. In our evaluation we look at the cumulative rebuffering time a user experienced during a simulation run.

\subsubsection{Prebuffering duration (start-up delay)}
At the start of a streaming session, user's video buffer is empty, so he has to wait until enough video data is downloaded to start playback. In contrast to rebuffering, a user at this state typically do not have information about network conditions. Especially when a user frequently starts a new streaming session, e.g., by switching TV channels, or when he repeatedly watches short videos, even a moderate start-up delay might severely degrade \ac{QoE} and even make the user decide not to watch the video at all.

\subsubsection{Mean media bit rate}
The mean video quality is obviously a factor that dramatically influences the overall \ac{QoE}, although studies have shown that it cannot be considered as a standalone \ac{QoE} measure for adaptive video streaming. In our evaluation, we identify mean video quality with mean media bit rate during a simulation run.

\subsubsection{Media bit rate fluctuations}
Recent studies have shown the significance of temporal quality fluctuations on the overall \ac{QoE}. As a measure of quality fluctuations we use the fraction of segments played out in a different quality than the preceding segment. That is, a quality fluctuations measure of 1\% means that one segment out of 100 was played in a different quality than its predecessor. For a segment duration of 2 seconds, e.g., this would mean one quality adaptation in 200 seconds.

Note, however, that due to significant media bit rate fluctuations within a single representation, resulting from \ac{VBR} encoding used by most modern video compression technologies, quality fluctuations cannot always be completely avoided.

\subsubsection{Media bit rate fairness}
Although a selfish user might not care about that, from the perspective of a service provider and/or system designer, a fare distribution of \ac{QoE} among clients sharing a common bandwidth resource is essential for the overall system performance. In our evaluation, we computed video quality fairness as follows. Taking the set of mean video qualities of all users in one simulation run, the fairness index is defined as the interquartile range, that is, the distance between the 0.25 and the 0.75 quantiles.

\subsection{Evaluation settings}\label{sec:evaluation_settings}

In this section we describe general settings such as video and network parameters.

\begin{figure}
\centering
\includegraphics{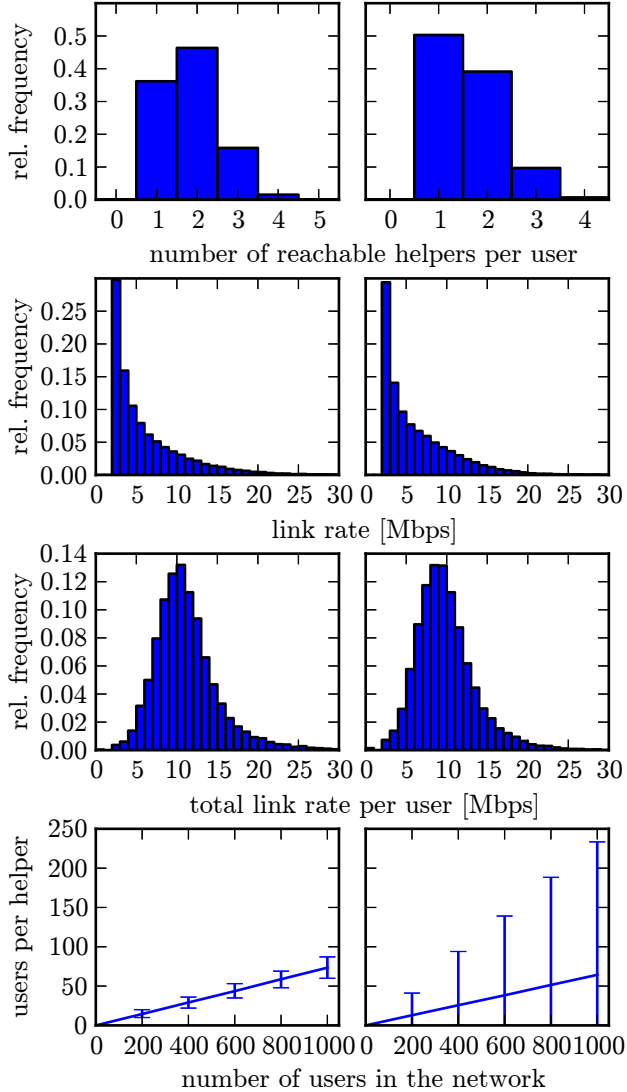}
\caption{Connectivity statistics for uniformely distributed users (left) and clustered users (right). From top to bottom: number of helpers per client; active link rates; sum of link rates per client; average number of users per helper for different total numbers of users in the network (with 10th and the 90th percentiles). See Section~\ref{sec:network_settings} for details.}
\label{fig:connectivity_statistics}
\end{figure}

\subsubsection{Video related settings}

For evaluation, we used a mix of 6 videos, contributed by the University of Klagenfurt~\cite{Lederer2012}. For each video, we selected 6 representations, ranging from approximately 500 kbps to 4.5 Mbps. Segment duration was 2 seconds.

The target buffer level of the video clients was set to 20 s. This value was shown to be sufficient to provide good performance in wireless networks in~\cite{Miller2013a}.

\subsubsection{Network related settings}\label{sec:network_settings}

The simulated network spans an area of 50 x 50 meters, covered by 25 helpers, distributed on a uniform grid. The duration of a scheduling time slot is set to 10 ms.

The path loss coefficients $g_{hu}(t)$ between helper $h$ and user $u$ are based on the WINNER II channel model~\cite{Kyosti2007}:
\begin{equation*}
g_{hu}(t)=10^{-0.1\text{PL}(d_{hu}(t))}\,,
\end{equation*}
where $d_{hu}(t)$ is the distance from helper $h$ to user $u$ at time $t$, and where 
\begin{equation}\label{eq:winner2.2}
PL(d) = A\log{d} + B + C\log{0.25 f_0} + \chi_{\text{dB}}\,.
\end{equation}
In~\eqref{eq:winner2.2}, $d$ is expressed in meters, the carrier frequency $f_0$ in GHz, and $\chi_{\text{dB}}$ denotes a shadowing log-normal variable with variance $\sigma_{\text{dB}}^2$. The parameters $A$, $B$, $C$ and $\sigma_{\text{dB}}^2$ are scenario-dependent constants. Among the several models specified in WINNER II we chose the A1 model~\cite{Kyosti2007}, representing an indoor small-cell scenario. In this case, $3\leq d\leq 100$, and the model parameters are given by $A=18.7$, $B=46.8$, $C=20$, $\sigma_{\text{dB}}^2=9$ in \ac{LOS} condition, or $A=36.8$, $B=43.8$, $C=20$, $\sigma_{\text{dB}}^2=16$ in \ac{NLOS} condition. For distances less than 3 m, we extended the model by setting $\text{PL}(d)=\text{PL}(3)$. Each link is in \ac{LOS} or \ac{NLOS} independently and at random, with a distance-dependent probability $p_l(d)$ and $1-p_l(d)$, respectively, where
\begin{equation*}
p_l(d)=\begin{cases}
1\hspace{5.1cm}\text{if}\;d\leq 3\;\text{m}\\
1-0.9(1-(1.24-0.6\log{d})^3)^{1/3}\;\;\text{otherwise}\,.
\end{cases}
\end{equation*}
In the following experiments, $d$ is updated in every scheduling time slot, while the random components, $\chi_{\text{dB}}$ and $p_l$ are updated every 5 seconds (except when $d$ falls below $3$ m, which forces the link to switch into the \ac{LOS} mode immediately, to maintain a consistent setting). 

Finally, links with a link rate below 2 Mbps in a particular scheduling time slot are not used for transmission. We call the remaining links \emph{active} links.

To visualize the resulting network conditions, Figure~\ref{fig:connectivity_statistics} shows some connectivity statistics. The left column shows statistics for users uniformly distributed over the simulated area, while the right column illustrates the case of clustered users, where users' distance to the center of the simulated area is exponentially distributed with $\lambda=0.2\ln{4}$. The top subfigures show the histogram of the number of helpers that can serve a client at a particular location. The subfigures in the second row show the histogram of link rates over all active links. The subfigures in the third row show the histogram of the sum of link rates for a client. The bottom subfigures show the average number of users served by a helper, for different total numbers of users in the network, plus 10th and 90th percentiles.

\subsubsection{Controller parameters}

All experiments are performed with different controller parameters, and different numbers of users. 
In all experiments, $K_d$ is set to 0. This is common practice in many applications. The reason behind it is that when the derivative of the system state is estimated from sampled measurements and the measurements are noisy, the derivative action amplifies the noise, introducing additional jitter in the system variable. In the following, whenever appropriate, results are only provided for selected parameter values, to improve the readability of the paper.

For better illustration, we would like to give an intuition for the scale of the parameters. For $K_p=-0.05$, if the buffer level is below the target value by 20 seconds (as, e.g., at the beginning of a streaming session), and the integral gain is at its equilibrium value (which is 1 in our case), then the client would try to download the video at twice the playback speed. That is, in one second the client would try to download two seconds of video, which would get him closer to the target value by 1 second. For a deviation of 10 seconds, the download rate would be 150\% of the playback rate, and so on. 


All simulations were repeated 30 times.

\begin{figure}
\centering
\includegraphics{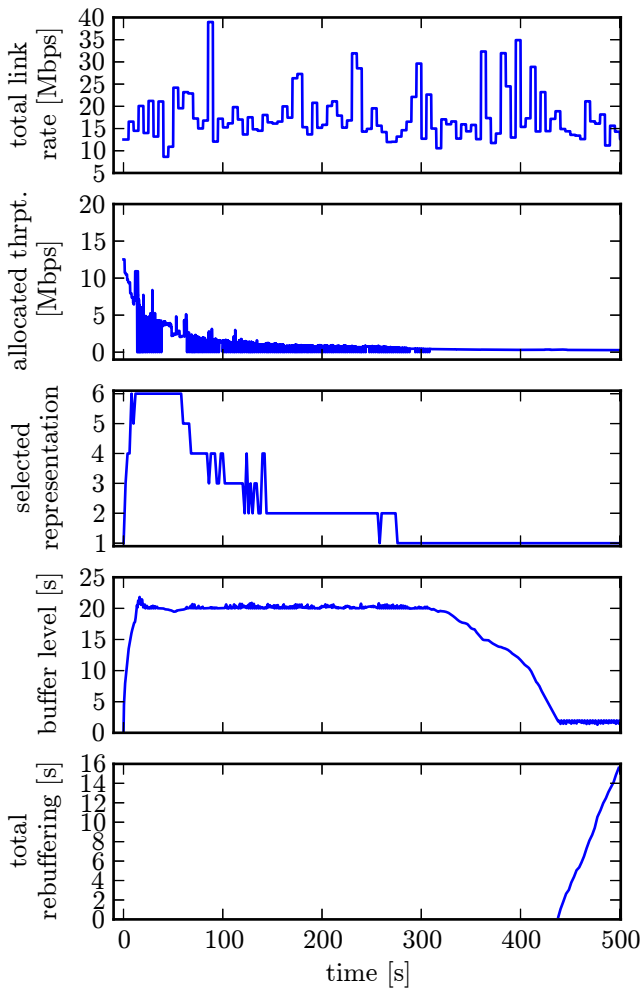}
\caption{Dynamics of one user during an example run of experiment 3. See Section~\ref{sec:evaluation_results} for details.}
\label{fig:example_users}
\end{figure}

\begin{figure*}
\centering
\subfloat[][]{
\includegraphics{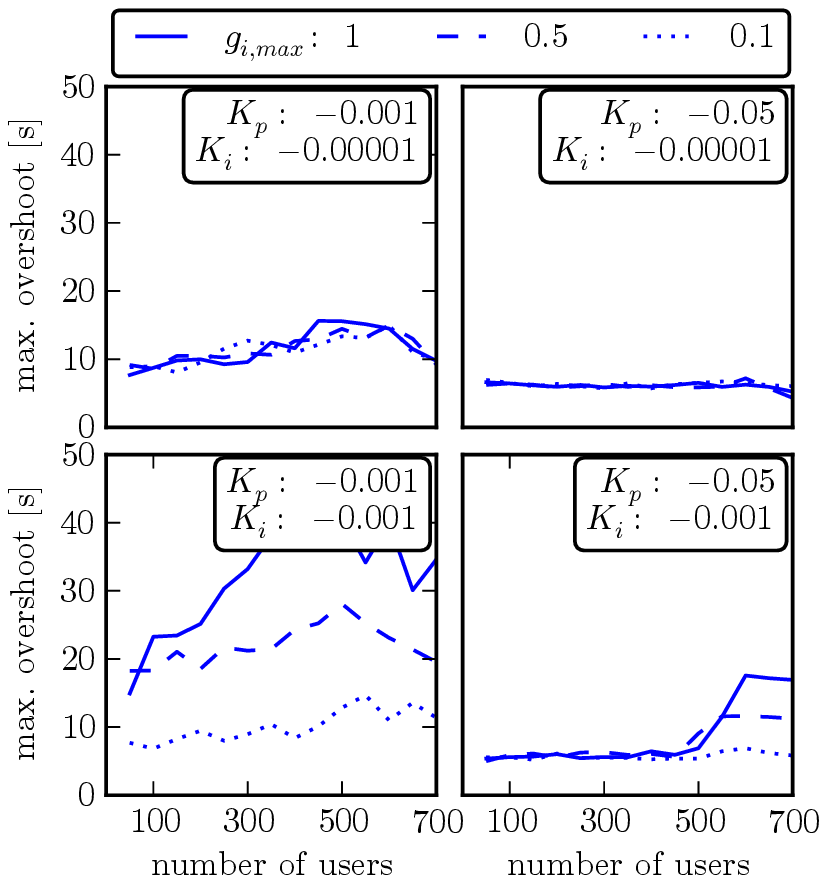}
\label{fig:bl_max_clustered}}
\subfloat[][]{
\includegraphics{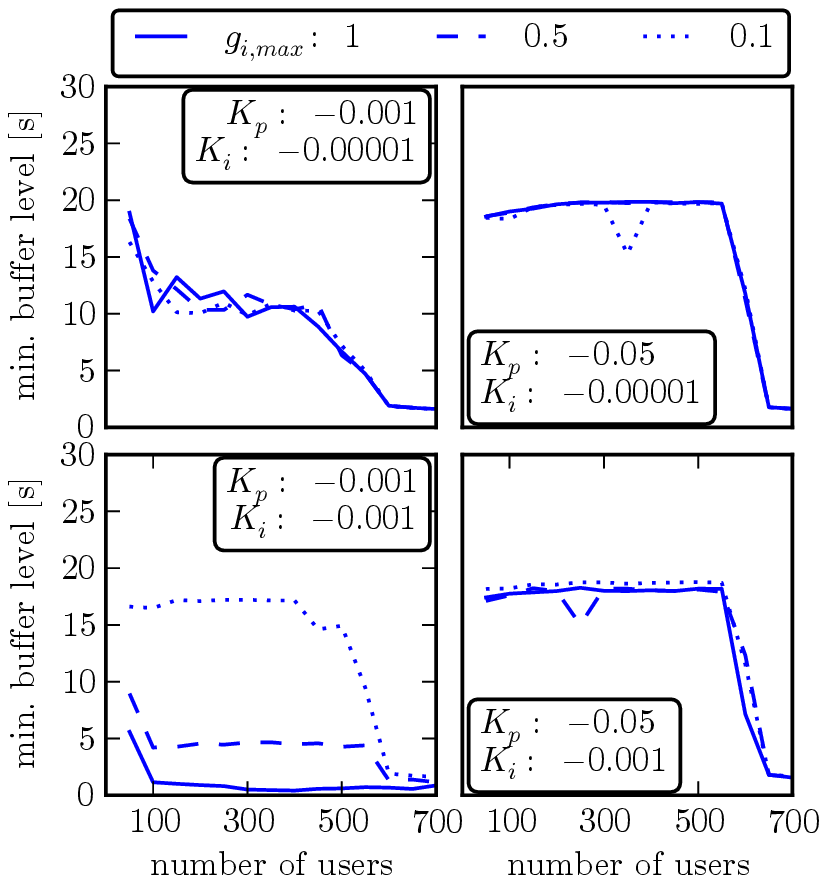}
\label{fig:bl_min_clustered}}
\caption{Stability analysis based on buffer level statistics from experiment 1, see Section~\ref{sec:stability} for details. (a) Maximum buffer level overshoot (difference between buffer level and target buffer level). (b) Minimum buffer level.}
\end{figure*}

\subsection{Experimental design}\label{sec:experimental_design}

We evaluate the proposed system in three types of experiments, each of them focusing on certain deployment scenarios, such as long-term users with no user churn, short-term users with high user churn, and a mix of short-term and long-term users. The performance is then compared with the performance of the baseline approach, described below.

Each of the experiments is executed with two different user distributions, called uniform and clustered in the following. With the former, arriving users are dropped at a random location, uniformly distributed over the simulated area. With the latter, arriving users are clustered around the center of the simulated area, with exponentially distributed distance ($\lambda=0.2\ln{4}$).

Upon joining the network, each user starts to watch a randomly selected video from a random point within the video. If he arrives at the end of the video, he continues to watch from the beginning. 


\subsubsection{Experiment 1}

The first experiment is intended to analyze system's behavior under constant load (fixed number of users) and without user churn (all users are long-term users). In particular, all users arrive during an initial arrival phase at the begin of the simulation and remain in the network until its end. Thus, after the arrival phase, the number of users in the network remains constant. 

The initial arrival phase starts at $t=0$. The users arrive at a rate of 10 users per second until a predefined number of users is reached. Then, the simulation continues with a constant number of users for another 400 seconds. 

\subsubsection{Experiment 1*}

In order to compare performance with a baseline approach, we rerun experiment 1 with the following transmission scheduling and quality selection. In each scheduling timeslot, every user is receiving data form exactly one helper, namely the one with the highest \ac{SINR}. (This is not always the closest helper, due to the random component in the pathloss.) Further, the client selects the video representation with the highest media bit rate that is still below the average throughput from the last 5 seconds.

\subsubsection{Experiment 2}

In this setting, the goal is to analyze system's performance under continuous user churn. That is, users continuously join and leave the network. As in experiment 1, there is an initial arrival phase, during which users arrive at a rate of 10 users per second until a certain number of users is reached. After that, 'old' users leave the network, while new users join it, at a rate of 2 users per second. 





\subsubsection{Experiment 3}

This experiment is intended to study system's behavior under constantly increasing load, in a deployment scenario with both short-term and long-term users. Here, new users continuously join the network at a constant rate of 2 users per second, and remain active until the end of the simulation run.

\subsection{Evaluation results}\label{sec:evaluation_results}

In this section, we present evaluation results for the four types of experiments, described in Section~\ref{sec:experimental_design}. We split the results according to the considered performance metrics: stability, rebuffering duration, start-up delay, mean media bit rate, media bit rate fluctuations, and media bit rate fairness.

Before we look into results for individual metrics, we would like to illustrate system behavior based on one example run. Figure~\ref{fig:example_users} shows dynamics of one user during a run of experiment 3, with $K_p=-0.05$, $K_i=-0.00001$, $K_d=0$, $g_{i,\max}=0.1$. The experiment runs for 500 seconds, that is, in the end there are 1000 users in the network. The plot shows the first user in the network, who starts to stream at second 0. The top subfigure shows the accumulated link rate to all neighboring helpers. The second subfigure shows network throughput allocated to the user. In scheduling timeslots that corresponds to inter-request delays, no resources are allocated to the user and thus his throughput drops to 0. The third subfigure shows the selected video representation, the fourth subfigure shows the buffer level. Finally, the bottom subfigure shows the total time spent in rebuffering.

\subsubsection{Stability}\label{sec:stability}

One of the issues that needs to be taken care of when designing a closed-loop controller is system's stability. In order to study stability, we analyze buffer level statistics of each user during each of the simulation runs. In particular, we look at maximum and minimum buffer levels, where the maximum and minimum operators are first applied to traces of individual users, then to resulting per user values, and finally to the whole set of runs for a specific configuration. In the following, we present results for experiment 1, with a uniform distribution of users across the simulated area. Results for other settings, omitted here for the lack of space, are consistent with presented findings. 
Finally, in order for the results not to be biased by system behavior during the initial arrival phase, we remove the initial arrival phase and the subsequent 100 seconds from each trace. 

Figure~\ref{fig:bl_max_clustered} shows the maximum buffer level overshoot (that is, the maximum difference between user's buffer level and the target buffer level). As expected, if the integral gain coefficient is large, as compared to the proportional gain coefficient, the system tends to become unstable. At the same time, however, we observe that the conditional integration anti-windup technique successfully combats this effect, if $g_{i,\max}$ is sufficiently small. Also the minimum buffer level, depicted in Figure~\ref{fig:bl_min_clustered}, confirms the efficiency of the conditional integration technique in avoiding system instability. 
We also studied the mean buffer level. We observed that it is within few seconds of the target buffer level, even for unstable configurations, and omit it here.


In the following, we only report results for $K_p=-0.05$, $K_i=-0.00001$, and $g_{i,\max}=0.1$, which we confirmed to be a stable configuration, and omit results for other configurations.    

\subsubsection{Rebuffering}\label{sec:rebuffering}

\begin{figure}[t!]
\centering
\includegraphics{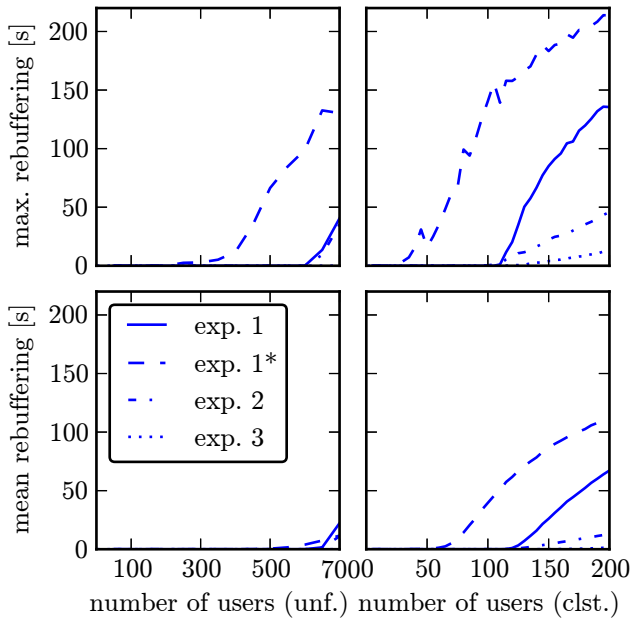}
\caption{Mean and maximum total rebuffering per user. See Section~\ref{sec:rebuffering} for details.}
\label{fig:rebuffering}
\end{figure}

One of the main factors influencing \ac{QoE} for video streaming is the amount of time a client spends rebuffering video data while the playback is halted. This happens when the playback buffer has been drained and the next segment does not arrive before its playback deadline. 

Figure~\ref{fig:rebuffering} shows mean and maximum total rebuffering per user, where the mean (maximum) is first taken over all users of a simulation run, and then over all runs performed for a setting. As in the previous section, we remove the initial arrival phase and the subsequent 100 seconds from each trace, for experiments 1, 1*, and 2 (experiment 3 does not contain an initial arrival phase).

As expected, we observe that the number of users the network can accomodate without rebuffering is higher with a uniform distribution of users. Moreover, we observe that the baseline approach results in significantly higher rebuffering values. 






\subsubsection{Prebuffering (start-up delay)}{\label{sec:prebuffering}

Another critical factor influencing \ac{QoE} is the prebuffering duration, or start-up delay. Especially in a mobile context, when users tend to watch shorter videos, long start-up delays can be very annoying.


Figure~\ref{fig:prebuffering} shows mean and maximum prebuffering delays for experiments 2 and 3, for uniform and clustered user distributions. For experiment 2, only users who arrived after the initial phase are considered. For experiment 3, to facilitate comparison, the x-axis shows number of users based on the user arrival rate of 2 users per second, instead of time.

With $K_p=-0.05$ and a target buffer level of 20 seconds, new users try to download the first segment at twice the media bit rate, that is, within one second. 
When there are few users in the system, the network can satisfy the corresponding throughput requests and even allocate some additional capacity to the individual users. When there are too many users in the system, the network cannot allocate the requested capacity for every user. When the load is moderate, many user receive exactly the requested capacity, resulting in one second prebuffering delay, as can be seen in Figure~\ref{fig:prebuffering}, left column.

\subsubsection{Mean media bit rate}\label{sec:mean_quality}

\begin{figure}[t!]
\centering
\includegraphics{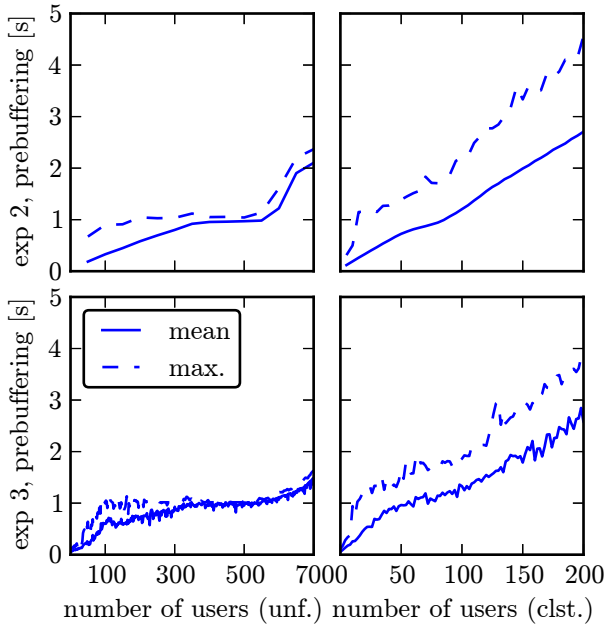}
\caption{Mean and maximum prebuffering duration (start-up delay) for experiments 2 (top) and 3 (bottom). See Section~\ref{sec:prebuffering} for details.}
\label{fig:prebuffering}
\end{figure}

Figure~\ref{fig:mean_quality} (top subfigure) illustrates mean video quality, represented by mean media bit rate, across all users. It shows results for all four experiments, for uniform (left column) and clustered (right column) user distributions. 

As in previous sections, the arrival phase was removed for experiments 1, 1*, and 2. In addition, the first 30 seconds of each user's trace were removed from experiments 2 and 3, since a user always starts to stream at lowest available quality. Because of the latter, results for experiments 2 and 3 do not include settings with less than 60 users, where no user remains in the network longer than 30 seconds.

We observe that experiment 1 and 1* exhibit comparable average media bit rates for both user distributions. With clustered user distribution, the controller driven approach offers slightly better values, especially for small numbers of users. 


\subsubsection{Media bit rate fluctuations}\label{sec:quality_fluctuations}

\begin{figure}[t!]
\centering
\includegraphics{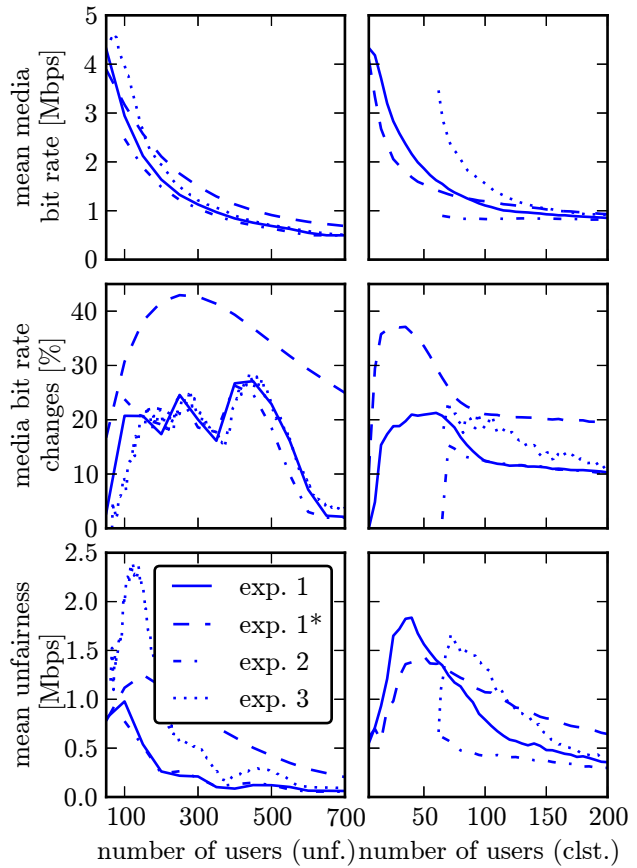}
\caption{From top to bottom: mean media bit rate in [Mbps]; mean percentage of segments played in a different representation than their predecessor (quality fluctuations indicator); mean interquartile range of media bit rate in [Mbps] (unfairness indicator). See Sections~\ref{sec:mean_quality}, \ref{sec:quality_fluctuations}, and~\ref{sec:quality_fairness} for details.}
\label{fig:mean_quality}
\end{figure}

Several studies have shown that severe quality fluctuations can dramatically degrade quality of experience even if the average quality is high. Especially in cases where network conditions may change very fast, such as in wireless networks, video clients have to implement adaptation strategies that avoid to immediately adapt video quality to dynamically varying throughput but that only react to long-term throughput changes.

The middle subfigure in Figure~\ref{fig:mean_quality} shows media bit rate fluctuations for all four experiments. 
As described in Section~\ref{sec:performance_metrics}, we measure media bit rate fluctuations as percentage of segments that were played in a different quality than their predecessor. With this definition, a value of 1\% means that an adaptation takes place every 100 segment. With a segment duration of 2 seconds, this corresponds to one adaptation in 200 seconds.

We observe that the amount of adaptations is up to twice as high with the baseline approach as with the controller driven approach.




\subsubsection{Media bit rate fairness}\label{sec:quality_fairness}

Finally, the bottom subfigure in Figure~\ref{fig:mean_quality} illustrates fairness by showing the interquartile range over all users in a network of their respective mean media bit rates. The higher the value, the lower the fairness. We observe that, except for the case of clustered users with less than 60 users in the network, the controller driven approach offers significantly better fairness than the baseline approach.



\section{Conclusion}\label{sec:conclusion}

In this study, we presented an approach for joint transmission scheduling and video quality selection in small-cell networks. The core of the approach is a \ac{PID} controller, known for its simplicity, analytical tractability, and robustness in the presence of modeling uncertainties and external disturbances. Although the controller in the studied setting is subject to saturation due to bandwidth constraints and constraints on available media bit rates, we successfully apply an anti-windup technique called conditional integration to stabilize the system. We additionally apply several heuristics that allow us to decentralize the designed mechanism, and specifically address issues that are known to affect \ac{QoE}.

We evaluated the performance of our approach under different conditions, including uniformely distributed and clustered users, as well as different mixes of short-term and long-term users, and settings with high user churn. The evaluation results showed that the developed approach performs well, outperforming the baseline approach.

Future work includes development of an outer control loop for the target buffer level that might be adjusted based on user's individual link quality statistics.

\acrodef{ADS}{Adobe Dynamic Streaming}
\acrodef{AP}{Access Point}
\acrodef{BOWL}{Berlin Open Wireless Lab}
\acrodef{CDN}{Content Delivery Network}
\acrodef{DSL}{Digital Subscriber Line}
\acrodef{MPEG-DASH}{Dynamic Adaptive Streaming over HTTP}
\acrodef{GOP}{Group of Pictures}
\acrodef{HLS}{Apple HTTP Live Streaming}
\acrodef{HTML}{Hypertext Markup Language}
\acrodef{HTTP}{Hypertext Transfer Protocol}
\acrodef{HAS}{HTTP-Based Adaptive Streaming}
\acrodef{IVP}{Initial Value Problem}
\acrodef{LAN}{Local Area Network}
\acrodef{LMI}{Linear Matrix Inequality}
\acrodef{LOS}{Line-of-Sight}
\acrodef{LQR}{Linear Quadratic Regulator}
\acrodef{LTE}{Long-Term Evolution}
\acrodef{MAC}{Media Access Control}
\acrodef{MANET}{Mobile Ad-Hoc Network}
\acrodef{MCNKP}{Multiple-Choice Nested Knapsack Problem}
\acrodef{MIMO}{Multiple Input Multiple Output}
\acrodef{MPC}{Model Predictive Control}
\acrodef{MPD}{Media Presentation Description}
\acrodef{MSS}{Microsoft SmoothStreaming}
\acrodef{NAT}{Network Address Translation}
\acrodef{NCS}{Networked Control System}
\acrodef{NLOS}{Non-Line-of-Sight}
\acrodef{NUM}{Network Utility Maximization}
\acrodef{ODE}{Ordinary Differential Equation}
\acrodef{OFDM}{Orthogonal Frequency-Division Multiplexing}
\acrodef{PI}{Proportional-Integral}
\acrodef{PID}{Proportional-Integral-Derivative}
\acrodef{QoE}{Quality of Experience}
\acrodef{QoS}{Quality of Service}
\acrodef{RTT}{Round-Trip Time}
\acrodef{TCP}{Transmission Control Protocol}
\acrodef{SDMA}{Spatial Division Multiple Access}
\acrodef{SINR}{Signal-to-Interference-plus-Noise Ratio}
\acrodef{SISO}{Single Input Single Output}
\acrodef{SMC}{Sliding Mode Control}
\acrodef{UDP}{User Datagram Protocol}
\acrodef{URL}{Uniform Resource Locator}
\acrodef{VBR}{Variable Bit Rate}
\acrodef{VoD}{Video on Demand}
\acrodef{VSC}{Variable Structure Control}
\acrodef{WLAN}{Wireless Local Area Network}
\acrodef{XAP}{Silverlight Application Package}

\bibliographystyle{IEEEtran}
\bibliography{IEEEabrv,library}

\end{document}